\documentclass[a4paper,UKenglish,cleveref, autoref, thm-restate]{lipics-v2021}
\usepackage{mathtools}
\usepackage[T1]{fontenc}
\usepackage{graphicx}
\usepackage{latexsym}
\usepackage{amsfonts}
\usepackage{amsmath,amssymb}
\usepackage{hyperref}
\usepackage{fix-cm}
\usepackage[linesnumbered,lined,ruled,noend,vlined]{algorithm2e}
\usepackage{cleveref}
\usepackage{tabularx,environ}
\usepackage{comment}
\usepackage{url}
\usepackage[full]{complexity}
\usepackage{lscape}
\usepackage{color}
\usepackage[normalem]{ulem}
\usepackage{todonotes}
\usepackage{xspace}
\usepackage{tikz}
\usepackage[appendix=inline]{apxproof}
\usepackage{cite}

\newcommand{\set}[1]{\{#1\}}

\newcommand{\inset}[2]{\{#1 | #2\}}

\newcommand{\size}[1]{| #1 |}
\newcommand{\order}[1]{O\hspace{-0.06cm}\left(#1\right)}

\newtheorem{problem}[lemma]{\bf Problem}

\Crefname{algocf}{Algorithm}{Algorithms}

\makeatletter
\@ifpackageloaded{algorithm2e}{
\SetKwProg{Fn}{Function}{}{}
\SetKwProg{Procedure}{Procedure}{}{}
\SetKwProg{Subprocedure}{Subprocedure}{}{}
\SetKwComment{tcc}{//}{}
\SetKw{Continue}{continue}
\SetKwFunction{Output}{Output}%
\SetKwFunction{EnumMC}{EMC}
\SetKwInOut{AlgInput}{Input}
\SetKwInOut{AlgOutput}{Output}
\SetKwInOut{AlgPrecondition}{Pre-conditions}
\SetKwInOut{AlgInvariant}{Invariants}
}{}
\makeatother

\newcommand{\ms}[1]{{\ifthenelse{\equal{#1}{}}{{ ms}}{{ ms}\xspace\left(#1\right)}}} 
\newcommand{\li}[1]{{\ifthenelse{\equal{#1}{}}{\ell}{\ell\xspace\left(#1\right)}}} 
\newcommand{\posms}[1]{{\ifthenelse{\equal{#1}{}}{{ p}}{{ p}\xspace\left(#1\right)}}} 

\newcommand{\comp}[1]{{\ifthenelse{\equal{#1}{}}{{\tt \mu}}{{\tt \mu}\xspace\left(#1\right)}}}
\newcommand{\mcc}[2]{{\ifthenelse{\equal{#1}{}}{{\tt mcc}}{{\tt mcc}\xspace\left(#1, #2\right)}}}
\newcommand{\minimize}[1]{{\ifthenelse{\equal{#1}{}}{{\Gamma}}{{\Gamma}\xspace\left(#1\right)}}}

\newtheoremrep{thm}[theorem]{Theorem}
\newtheoremrep{lem}[lemma]{Lemma}
\newtheoremrep{obs}[observation]{Observation}

\crefname{obs}{observation}{observations}
\Crefname{obs}{Observation}{Observations}

\nolinenumbers

\title{On the Complexity of Hyperpath and Minimal Separator Enumeration in Directed Hypergraphs}
\titlerunning{On the Complexity of Hyperpath and Minimal Separator Enumeration}

\author{Kazuhiro Kurita}{Nagoya University, Nagoya, Japan}{kurita@i.nagoya-u.ac.jp}{https://orcid.org/0000-0002-7638-3322}{
This work is partially supported by JSPS KAKENHI Grant Numbers 
JP22H03549, 
JP25K21273, 
JP25K03080, and 
JP25K00136. 
}

\author{Kevin Mann}{Universit\"at Trier, Trier, Germany}{mann@uni-trier.de}{https://orcid.org/0000-0002-0880-2513}{[funding]}

\authorrunning{K. Kurita and K. Mann} 

\Copyright{Jane Open Access and Joan R. Public} 

\ccsdesc[500]{Mathematics of computing~Graph algorithms}

\keywords{Output-sensitive Enumeration, Directed hypergraph, $s$-$t$ hyperpath, Minimal $s$-$t$ separator, \NP-hardness, Polynomial-delay}

\category{} 

\relatedversion{} 

\acknowledgements{This work was supported by the Research Institute for Mathematical Sciences, an International Joint Usage/Research Center located at Kyoto University. We would like to thank the anonymous reviewers for their valuable comments and suggestions, which helped improve the quality of this paper.}

\begin{document}

\maketitle              
\begin{abstract}
In this paper, we address the enumeration of (induced) $s$-$t$ paths and minimal $s$-$t$ separators.
These problems are some of the most famous classical enumeration problems that 
can be solved in polynomial delay by simple backtracking for a (un)directed graph.
As a generalization of these problems, we consider the (induced) $s$-$t$ hyperpath and minimal $s$-$t$ separator enumeration in a \emph{directed hypergraph}.
We show that extending these classical enumeration problems to directed hypergraphs drastically changes their complexity.
More precisely,
there are no output-polynomial time algorithms for the enumeration of induced $s$-$t$ hyperpaths and minimal $s$-$t$ separators unless $\P = \NP$, and 
if there is an output-polynomial time algorithm for the $s$-$t$ hyperpath enumeration, then the minimal transversal enumeration can be solved in output polynomial time even if a directed hypergraph is $BF$-hypergraph.
Since the existence of an output-polynomial time algorithm for the minimal transversal enumeration has remained an open problem for over 45 years, 
it indicates that the $s$-$t$ hyperpath enumeration for a $BF$-hypergraph is not an easy problem.
As a positive result, the $s$-$t$ hyperpath enumeration for a $B$-hypergraph can be solved in polynomial delay by backtracking.
\end{abstract}
\newpage
\section{Introduction}
Enumeration
 problems on graphs have been widely studied, and in particular, 
the complexity of enumerating minimal subsets that satisfy connectivity constraints has been widely studied, i.e., (induced) $s$-$t$ paths~\cite{Uno:DS:2014,ReadTarjan}, minimal $s$-$t$ cuts and multiway cuts~\cite{10.1145/3517804.3524148,DBLP:conf/mfcs/KuritaK20}, minimal $s$-$t$ separators~\cite{Takata:DAM:2010,DBLP:journals/ijfcs/BerryBC00}, minimal strongly connected spanning subgraphs~\cite{DBLP:journals/algorithmica/KhachiyanBEG08}, and minimal Steiner trees and its variants~\cite{KIMELFELD2008335,10.1145/3517804.3524148}.
When evaluating the complexity of enumeration algorithms, 
we sometimes use an \emph{output-sensitive} manner~\cite{JOHNSON1988119}.
An enumeration algorithm is called an \emph{output-polynomial time} algorithm if 
the total running time is bounded by $\order{\mathrm{poly}(n + N)}$ time, where $n$ is the size of an input and $N$ is the number of outputs.
In enumeration algorithms, \emph{delay} is used as an efficiency measure in addition to the total computation time.
The delay is defined as the maximum time until the first solution is output, the time between outputting the $i$-th and $i+1$-th solutions, and the time between outputting the last solution and termination.
If the delay is bounded by $\order{\mathrm{poly}(n)}$, we call it a \emph{polynomial-delay algorithm}.

As a generalization of directed graphs, \emph{directed hypergraphs} are known~\cite{AUSIELLO2017293,GALLO1993177}.
In a directed graph, an arc consists of a single tail and a single head.
In contrast, in a directed hypergraph, an hyperarc is defined by a disjoint pair of sets of vertices.
One set of vertices represents the \emph{tails} and the other set of vertices represents the \emph{heads}.
In particular, a directed hypergraph in which the head of every hyperarc has cardinality one is called a \emph{$B$-hypergraph},
and a directed hypergraph in which either the head or the tail of every hyperarc has cardinality one is called a \emph{$BF$-hypergraph}.
In addition, as a generalization of reachability in directed graphs, the concept of $B$-connectivity is defined for directed hypergraphs.
$B$-connectivity relates to forward chaining, namely, a vertex $t$ is $B$-connected from a vertex $s$ if and only if the transitive closure/forward chaining of $s$ reaches $t$.

A $B$-connectivity on a directed hypergraph is a mathematical model that can represent a broader range of problems than a connectivity on a directed graph.  
It is used to model various applications such as chemical reaction networks~\cite{doi:10.1089/cmb.2023.0242,OZTURAN2008881}, the formulation of the maximum Horn \SAT{} problem~\cite{Gallo1998,GALLO1993177}, transportation networks~\cite{VOLPENTESTA2008390,Gallo1998_model}, and conflict-free Petri nets~\cite{ALIMONTI2011320,10.1007/3-540-56402-0_55,10.1007/BFb0023812}.  
By \cite{10.1007/BFb0023812}, there is connection between reachability in conflict-free petri net and paths on $B$-hypergraphs\footnote{The objects called hypergraphs in \cite{10.1007/BFb0023812}, are called $B$-hypergraphs in this paper.}.

While directed hypergraphs can model a wider variety of phenomena than directed graphs, optimization problems on directed hypergraphs become \NP-complete, even for problems that can be solved in polynomial time on directed graphs.
Examples of problems that become \NP-complete on directed hypergraphs include finding a directed cycle~\cite{OZTURAN2008881},~\footnote{Note that there are multiple definitions of a directed cycle in a directed hypergraph. The definition used in the paper~\cite{OZTURAN2008881} is given in Section 2.}
a minimum $s$-$t$ cut~\cite{Gallo1998}, the strongly connected components~\cite{Allamigeon2014}, and a shortest path~\cite{AUSIELLO2017293,10.1007/978-3-642-32147-4_1}. 
See the following survey for more details~\cite{AUSIELLO2017293}.

Given these facts, it is evident that when extending many optimization problems to directed hypergraphs, they tend to be intractable.
We aim to investigate how the complexity of fundamental enumeration problems changes when extending from directed graphs to directed hypergraphs.

Enumeration of (induced) $s$-$t$ path, minimal $s$-$t$ cuts, and minimal separators in a (directed) graph are fundamental problems in the field of enumeration algorithms, and many theoretically efficient enumeration algorithms have been developed~\cite{ReadTarjan,Provan:Algorithmca:1996,Birmele:Ferreira:SODA:2013,Takata:DAM:2010,Tsukiyama1980,Uno:DS:2014}.
It is well known that these problems can be enumerated with polynomial delay using a simple backtracking approach (also called the binary partition and the flashlight approach).

The goal of this paper is to show that enumerating induced $s$-$t$ hyperpaths and minimal $s$-$t$ separators in directed hypergraphs cannot be solved in  output-polynomial time 
unless $\P = \NP$.
Furthermore, we prove that
if there is an output-polynomial time algorithm for the $s$-$t$ hyperpath enumeration for $BF$-hypergraphs,
then the minimal transversal can be solved in output-polynomial time.
The existence of an output-polynomial time algorithm for minimal transversal enumeration has remained an open problem for over 45 years and
is one of the most famous open problems in the field of enumeration algorithms~\cite{EITER20082035}.
Finally, as a positive result, we show that when directed hypergraphs are restricted to $B$-hypergraphs, 
$s$-$t$ hyperpaths can be enumerated with polynomial delay using a backtracking approach.
Our algorithm solves a slightly more general problem than the $s$-$t$ hyperpath enumeration and 
it can be applied to minimal directed Steiner tree enumeration.
Moreover, this positive result shows that we can enumerate all minimal unsatisfiable subformula of a Horn \SAT{} formula in polynomial delay. 


\section{Preliminaries}
Let $V$ be a set of elements.
We define $\mathcal A$ as a set of a pair of disjoint subsets of $2^V$.
We call $\mathcal D = (V, \mathcal A)$ a \emph{directed hypergraph}.
An element of $V$ and $\mathcal A$ is called a \emph{vertex} and a \emph{hyperarc} of $\mathcal D$, respectively.
We denote the set of vertices and hyperarcs in $\mathcal D$ as $V(\mathcal D)$ and $\mathcal A(\mathcal D)$.
For a hyperarc $A = (T, H)$,
$T$ and $H$ are called the \emph{tails} of $A$ and \emph{heads} of $A$, respectively.
It is denoted by $T(A)$ and $H(A)$, respectively.
A hyperarc $A$ is called a \emph{$B$-hyperarc} (a \emph{$F$-hyperarc}) if $\size{H(A)} = 1$ ($\size{T(A)} = 1$).
Notice that $B$ and $F$ are abbreviations of ``backward'' and ``forward'', respectively.
A directed hypergraph $\mathcal D = (V, \mathcal A)$ is called a \emph{$B$-hypergraph} (a \emph{$BF$-hypergraph}) if any hyperarc in $\mathcal A$ is a $B$-hyperarc (a $B$-hyperarc or a $F$-hyperarc).
For a directed hypergraph $\mathcal D = (V, \mathcal A)$, 
a directed hypergraph $\mathcal F = (U, \mathcal B)$ is a \emph{subhypergraph of $\mathcal D$} if $U \subseteq V$ and $\mathcal B \subseteq \mathcal A$.
For a set of hyperarcs $\mathcal B \subseteq \mathcal A$, 
a directed hypergraph $\mathcal F = (U, \mathcal B)$ is 
an \emph{edge-induced subhypergraph of $\mathcal D$} induced by $\mathcal B$ if $U = \bigcup_{(T, H) \in \mathcal B}(T \cup H)$.
We denote it as $\mathcal D[\mathcal B]$.
A directed hypergraph $\mathcal F = (U, \mathcal B)$ is an \emph{induced subhypergraph of $\mathcal D$} induced by $U$ if $U \subseteq V$ and $B = \inset{A \in \mathcal A}{H(A) \subseteq U, T(A) \subseteq U}$.
We denote it as $\mathcal D[U]$.

In the theory of directed hypergraphs, 
there exist multiple definitions of connectivity and reachability.
We adopt $B$-connectivity as the notion of connectivity~\cite{THAKUR20092592,AUSIELLO2017293}.~\footnote{Gallo et al.~\cite{GALLO1993177} provides a characterization for making a vertex $t$ $B$-connected from a vertex $s$, but Nielsen et al. points out that it is incorrect~\cite{nielsen2001remark}.}
The \emph{$B$-connection} from a vertex $s$ is defined as follows: 
(i)  $s$ is $B$-connected from $s$, and 
(ii) if there is a hyperarc $A$ such that all the vertices in $T(A)$ are $B$-connected from $s$, then $H(A)$ is $B$-connected from $s$.
A \emph{$B$-hyperpath} from $s$ to a vertex $t$ is an inclusion-wise minimal subhypergraph such that $t$ is $B$-connected from $s$.
From the minimality of $B$-hyperpath, the set of vertices in an $s$-$t$ hyperpath
is uniquely defined by a set of hyperarcs.
Thus, we refer to a set of hyperarcs satisfying the condition as an $s$-$t$ hyperpath hereafter.
A set of hyperarcs $P_{st}$ is an $s$-$t$ hyperpath if and only if there is a sequence of hyperarcs $(A_1, \ldots, A_k)$ such that 
(i)   for each $A_i$, $T(A_i) \subseteq \set{s} \cup \bigcup_{1 \le j \le i-1}H(A_j)$,
(ii)  $t \in H(A_{k})$, and
(iii) for any proper subhypergraph of $P_{st}$, $t$ is not $B$-connected from $s$. See Section~2 in \cite{AUSIELLO200527} for more details.
We give an example of an $s$-$t$ hyperpath in \Cref{fig:example}.
A set of vertices $P_{st}$ is an \emph{induced $s$-$t$ hyperpath} if $t$ is $B$-connected from $s$ in $\mathcal D[P_{st}]$.
A set of vertices $X$ of $\mathcal D = (V, \mathcal A)$ is an \emph{$s$-$t$ separator} if $X$ does not contain $s$ and $t$ and
$t$ is not $B$-connected from $s$ in $D[V \setminus X]$.
An $s$-$t$ separator $X$ is \emph{minimal} if any proper subset of $X$ is not an $s$-$t$ separator.

\begin{figure}
    \centering
    \includegraphics[width=0.6\linewidth]{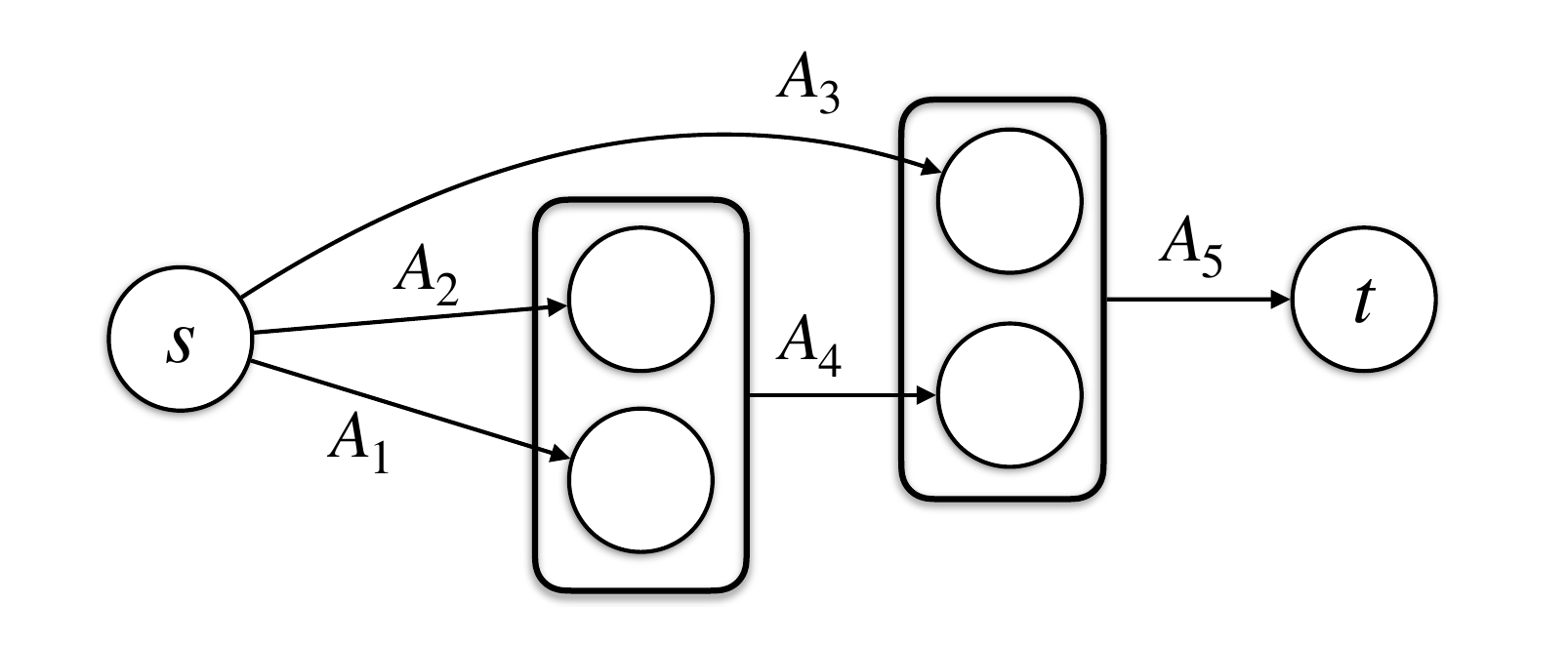}
    \caption{An example of a directed hypergraph and an $s$-$t$ hyperpath. If the cardinality of tail is more than one, we represent tails using a rectangle. It has an $s$-$t$ hyperpath with an order representation $(A_1, A_2, A_3, A_4, A_5)$.}
    \label{fig:example}
\end{figure}

In this paper, we deal with the following enumeration problems.
\begin{problem}[\textsc{(Induced) $s$-$t$ hyperpath Enumeration}]
    Let $\mathcal D=(V,\mathcal{A})$ be a directed hypergraph and $s$ and $t$ be vertices in $V$.
    The task is to enumerate all (induced) $s$-$t$ hyperpaths in $\mathcal D$.
\end{problem}

\begin{problem}[\textsc{Minimal $s$-$t$ Separator Enumeration}]
    Let $\mathcal D=(V,\mathcal{A})$ be a directed hypergraph and $s$ and $t$ be vertices in $V$.
    The task is to enumerate all minimal $s$-$t$ separators in $\mathcal D$.
\end{problem}

\section{Hardness Results}
We show three hardnesses in this section.
In \Cref{subsec:AIP} and \Cref{subsec:AMS}, 
we show that there are no output-polynomial time algorithms for enumerating all induced $s$-$t$ hyperpaths and minimal $s$-$t$ separators on $B$-hypergraphs , respectively, unless $\P = \NP$.
In \Cref{subsec:bpath:bf}, we show that 
\textsc{$s$-$t$ hyperpath Enumeration} for $BF$-hypergraphs is at least as hard as \textsc{Minimal Transversal Enumeration}.
More formally, if we have an output-polynomial time algorithm for \textsc{$s$-$t$ hyperpath Enumeration},
then minimal transversal enumeration can be solved in output polynomial time.

To show the hardness results in \Cref{subsec:AIP} and \Cref{subsec:AMS}, 
we show the \NP-hardness of \emph{another solution problem} (also called finished decision problem~\cite{https://doi.org/10.1002/mcda.1603}).
More formally, we consider the following problems.

\begin{problem}[\textsc{Another Induced $s$-$t$ hyperpath}]
    Let $\mathcal D$ be a directed hypergraph, $s$ and $t$ be vertices, and
    $\mathcal P_{st}$ be a set of induced $s$-$t$ hyperpaths.
    The task is to determine whether there is an induced $s$-$t$ hyperpath not contained in $\mathcal P_{st}$.
\end{problem}

\begin{problem}[\textsc{Another Minimal $s$-$t$ Separator}]
    Let $\mathcal D$ be a directed hypergraph, $s$ and $t$ be vertices, and
    $\mathcal X$ be a set of minimal $s$-$t$ separators.
    The task is to determine whether there is a minimal $s$-$t$ separator not contained in $\mathcal X$.
\end{problem}

If we have an output-polynomial time algorithm for these enumeration problems,
these another solution problems can be solved in polynomial time.
Such a proof technique has been used as folklore~\cite{JOHNSON1988119,DBLP:journals/ipl/BrosseDKLUW24,DBLP:journals/dam/BorosM24}.
For the sake of completeness, we explain why, if these another problems are \NP-complete, 
their corresponding enumeration problems cannot be solved in output-polynomial time unless $\P = \NP$.

Suppose that there is an output-polynomial time algorithm $\mathcal A$ for \textsc{Induced $s$-$t$ hyperpath Enumeration}.
We run the algorithm $\mathcal A$ at most $(\size{V(G)} + \size{\mathcal P_{st}})^c $ steps, where $c$ is some constant. 
If there is no induced $s$-$t$ hyperpaths that are not contained in $\mathcal P_{st}$, 
$\mathcal A$ terminates outputting $\mathcal P_{st}$.
Otherwise, $\mathcal A$ does not terminate.
In other words, if $\mathcal A$ does not terminate, 
it can be concluded that there exists an induced $s$-$t$ hyperpath not included in $\mathcal P_{st}$.
Therefore, if these another solution problems are \NP-hard, 
there is no output-polynomial time enumeration algorithm unless $\P = \NP$.
Note that this algorithm can answer only ``Yes'' or ``No'', but it cannot find an induced $s$-$t$ hyperpath not contained in $\mathcal P_{st}$.
However, since the aim of a decision problem is to answer ``Yes'' or ``No,''
this algorithm solves these another solution problems.

In \Cref{subsec:bpath:bf}, we show that $s$-$t$ hyperpath enumeration is at least as hard as \textsc{Minimal Transversal Enumeration}.
More formally, we show that for a hypergraph $\mathcal H = (V, \mathcal E)$, 
there is a bijection between $\mathrm{Tr}(\mathcal H)$ and the set of all $s$-$t$ hyperpaths in 
a directed hypergraph $\mathcal D(\mathcal H)$, where $\mathrm{Tr}(\mathcal H)$ is the set of minimal transversals in $\mathcal H$ and $\mathcal D(\mathcal H)$ is a directed hypergraph with $\order{\size{\mathcal E}}$ vertices and $\order{\size{\mathcal E}}$ edges.
A set of vertices $U$ is called a \emph{transversal} if for any $E \in \mathcal E$, 
$E\cap U \neq \emptyset$.
If for any $U' \subset U$, $U'$ is not a transversal, $U$ is a \emph{minimal transversal}.
\textsc{Minimal Transversal Enumeration} can be solved in output-quasi polynomial time~\cite{FREDMAN1996618}.
On the other hand, output-polynomial time enumeration of minimal transversals is a longstanding open problem in the field of enumeration algorithms~\cite{EITER20082035}.
Thus, our ``hardness'' result provides evidence that solving \textsc{$s$-$t$ hyperpath Enumeration} in output-polynomial time is not easy.

\subsection{Another Induced \texorpdfstring{$s$-$t$}{s-t} hyperpath}\label{subsec:AIP}
We give a polynomial-time reduction from 3-\SAT{}.
Let $\phi = C_1 \land \ldots \land C_m$ be a 3-CNF formula.
We denote the set of variables in $\phi$ as $V(\phi)$.
We construct a directed hypergraph $\mathcal D = (V, \mathcal A)$ from $\phi$ as follows.
We add two vertices $s$ and $t$ to $V$.
For each variable $x_i \in V(\phi)$, we add two vertices $x_i, \bar x_i$ to $V$.
For each clause $C_j$, we add a vertex $c_j$ to $V$.
We next define the set of hyperarcs.
For each $x_i$ and $\bar x_i$, we add three hyperarcs $(\set{x_i, \bar x_i}, \set{t})$, $(\set{s}, \set{x_i})$, and $(\set{s}, \set{\bar x_i})$.
For each clause $C_j = \ell^j_1 \lor \ell^j_2 \lor \ell^j_3$, we add three hyperarcs $(\set{\ell^j_1}, \set{c_j})$, $(\set{\ell^j_2}, \set{c_j})$, and $(\set{ \ell^j_3}, \set{c_j})$.
Finally, we add a hyperarc $(\set{c_1, \ldots, c_m}, \set{t})$.
We denote the resultant directed hypergraph as $\mathcal D_\phi$, and $\mathcal D_\phi$ is a $B$-hypergraph.
Notice that this reduction can be done in $\order{\mathrm{poly}(n + m)}$ time, where $n$ is the number of variables in $\phi$.

\begin{figure}[t]
    \centering
    \includegraphics[width=0.9\linewidth]{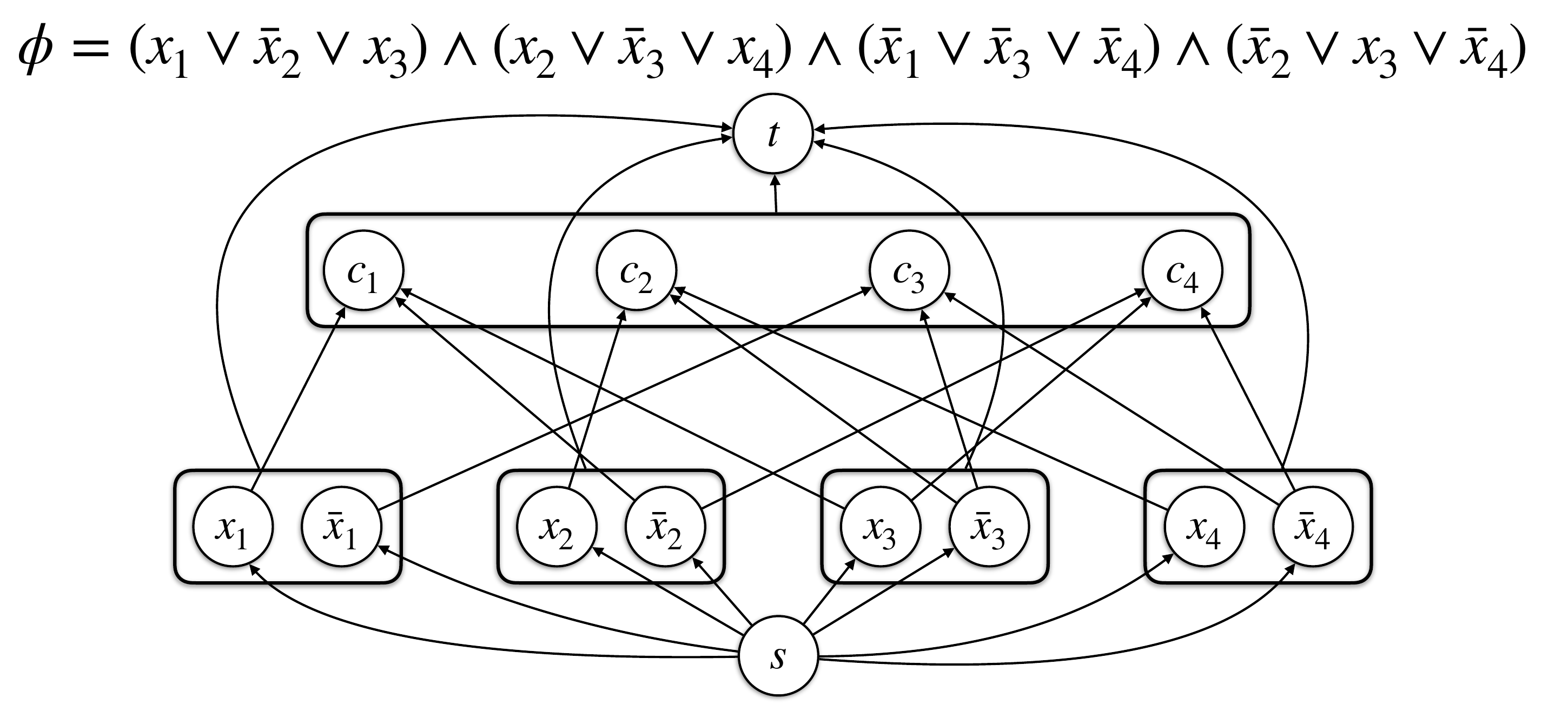}
    \caption{A reduction from 3-\SAT{} to \textsc{Another Induced $s$-$t$ hyperpath}. A directed hypergraph $\mathcal D_\phi$ has an induced $s$-$t$ hyperpath without 
    $\set{s, x_1, \bar x_1, t}$, 
    $\set{s, x_2, \bar x_2, t}$, 
    $\set{s, x_3, \bar x_3, t}$, and
    $\set{s, x_4, \bar x_4, t}$, if and only if $\phi$ is satisfiable.
    In this case, since
    $x_1 = 1$,
    $x_2 = 1$,
    $x_3 = 0$, and
    $x_4 = 0$ is a satisfying assignment, 
    $\mathcal D_\phi$ has an induced $s$-$t$ hyperpath $\set{s, x_1, x_2, \bar x_4, c_1, c_2, c_3, c_4, t}$.
    If we assign $x_1 = 1$, $x_2 = 1$, and $x_4 = 0$, then $\phi$ is satisfied regardless of the assignment to $x_3$. 
    That is, $x_3$ becomes a ``don't care'' variable, and therefore neither $x_3$ nor $\bar{x}_3$ is included in the hyperpath.
    }
    \label{fig:reduction-AIP}
\end{figure}

The resultant directed hypergraph $\mathcal D_\phi$ has an induced $s$-$t$ hyperpath with four vertices.
For each $1 \le i \le n$, $\set{s, x_i, \bar x_i, t}$
is an induced $s$-$t$ hyperpath.
We denote the set of such induced $s$-$t$ hyperpaths as $\mathcal P_\phi$.
\Cref{fig:reduction-AIP} gives an example of our reduction. 
We show that $\mathcal D_\phi$ has an induced $s$-$t$ hyperpath that is not included in $\mathcal P_{\phi}$ if and only if $\phi$ has a satisfying assignment.

\newcommand{\revised}[1]{\textcolor{red}{#1}}
\begin{lemma}\label{lem:reduction:AIP}
    For a 3-CNF formula $\phi$, 
    $\phi$ has a satisfying assignment if and only if 
    $\mathcal D_\phi$ has an induced $s$-$t$ hyperpath not contained in $\mathcal P_\phi$.
\end{lemma}
\begin{proof}
 Suppose that $\phi$ has a satisfying assignment $\alpha: V_\phi \to \set{0, 1}$.
    We define a set of vertices in $\mathcal D_\phi$ using $\alpha$.
    We define $U = V(\mathcal D_\phi)$.
    If $\alpha(x_i) = 0$, we remove a vertex $x_i$ from $U$, otherwise
    if $\alpha(x_i) = 1$, we remove a vertex $\bar x_i$ from $U$.
    We show that there is an induced $s$-$t$ hyperpath contained in $U$.    
    Since each clause $C_i = \ell^i_1 \lor \ell^i_2 \lor \ell^i_3$ is satisfied by $\alpha$, 
    at least one literal $\ell^i_j$ becomes $1$.
    Therefore, each $c_i$ is $B$-connected from $s$ in $\mathcal D_\phi[U]$.
    Since $\mathcal D_\phi[U]$ has a hyperarc $(\set{c_1, \ldots, c_m}, \set{t})$,
    $t$ is $B$-connected from $s$ in $\mathcal D_\phi[U]$ and there is an induced $s$-$t$ hyperpath contained in $U$.
    
    We show the opposite direction.
    Let $P$ be an induced $s$-$t$ hyperpath that is not contained in $\mathcal P_\phi$.
    For each pair $x_i$ and $\bar x_i$, 
    $x_i \notin P $ or $ \bar{x}_i \notin P$.
        If $P$ contains both $x_i$ and $\bar x_i$, it contradicts the minimality since 
        $t$ is $B$-connected from $s$ in $\mathcal D_{\phi}[\set{s, t, x_i, \bar x_i}]$.      
    If there is a vertex $c_j \not\in P$, 
    $\mathcal D_\phi[P]$ does not have a hyperarc $A$ that contains $t$ as the head.
    Therefore, it contradicts that $P$ is an induced $s$-$t$ hyperpath.
    From the above discussion, 
    $\set{c_1, \ldots, c_m} \subseteq P$ and 
    $x_i \notin P$ or $\bar x_i \notin P$ for each $1\le i \le n$.
    We define an assignment $\alpha_P$ as follows: 
    $\alpha_P(x_i) = 1$ if $x_i \in P$ and
    $\alpha_P(x_i) = 0$ if $x_i \not\in P$.
    We show that $\alpha_P$ satisfies $\phi$.
    For each clause $C_i = \ell^i_1 \lor \ell^i_2 \lor \ell^i_3$, 
    at least one literal $\ell^i_j$ becomes $1$ by $\alpha_P$.
    Otherwise, it contradicts that either $P$ contains $\set{c_1, \ldots, c_m}$ or $P$ is an induced $s$-$t$ hyperpath.
    Therefore, $\alpha$ is a satisfying assignment of $\phi$.
\end{proof}

\Cref{lem:reduction:AIP} implies there is no output-polynomial time algorithm for \textsc{Induced $s$-$t$ hyperpath Enumeration} unless $\P = \NP$.
By making minor modifications to our reduction, 
it is possible to show a stronger hardness result.
\textsc{Another induced $s$-$t$ hyperpath} is still \NP-hard even if the cardinality of a tail is at most two.
Intuitively, since $\mathcal D$ contains only one hyperarc with a large tail,
we can satisfy this condition by ``decomposing'' that hyperarc into multiple hyperarcs, each having a tail of cardinality two, in a structure resembling a binary tree.

Let $A$ be the hyperarc  $(\set{c_1, \ldots, c_m}, \set{t}) \in \mathcal A(\mathcal D_\phi)$.
To simplify the following discussion, we assume that $m$ is a power of two.
Notice that if $m$ is not a power of two, we add one dummy variable $z$ and dummy clauses $(z \lor z \lor z)$ to make $m$ a power of two.
We add vertices $d_2, \ldots, d_{m-1}$ to $\mathcal D$, and 
replace a hyperarc $(\set{c_1, \ldots, c_m}, \set{t})$ with hyperarcs
$\inset{(\set{w_{2i}, w_{2i+1}}, \set{w_{i}})}{1 \le i \le m-1}$, where we define $w_i$ as follows.
We define $w_1 = t$, 
if $i$ is at least $m$ and at most $2m-1$, $w_i = c_{i - m + 1}$, and otherwise $w_i = d_{i}$.
We denote the resultant graph as $\mathcal D'_\phi$.
Notice that the number of dummy clauses is at most $m$, this reduction works in polynomial time.
Any induced $s$-$t$ hyperpath contains all vertices in $\mathcal D'$, 
there is a bijection between the set of all induced $s$-$t$ hyperpaths in $\mathcal D_\phi$ and $\mathcal D'_\phi$.
Therefore, \Cref{lem:reduction:AIP} also holds even for $\mathcal D'_\phi$ and  the following theorem and corollary hold.

\begin{theorem}
    \textsc{Another Induced $s$-$t$ hyperpath} on a $B$-hypergraph $\mathcal D = (V, \mathcal A)$ is \NP-complete even if 
    $\size{\mathcal P_{st}} \le \size{V}$ and $\size{T(A)} \le 2$ for any $A \in \mathcal A$.
    Therefore, 
    there are no output-polynomial time algorithms for enumerating all induced $s$-$t$ hyperpaths in a $B$-hypergraph unless $\P = \NP$ even if $\size{T(A)} \le 2$ for any $A \in \mathcal A$.    
\end{theorem}

\subsection{Another Minimal \texorpdfstring{$s$-$t$}{st} Separator}\label{subsec:AMS}
We give a reduction from 3-\SAT{}.
This reduction is similar to the reduction in the previous subsection.
Let $\phi = C_1 \land \ldots \land C_m$ be a $3$-CNF formula.
We construct $\mathcal D_\phi = (V_\phi, \mathcal A_\phi)$ as follows.
We add vertices $s$, $t$, and $c$ to $V_\phi$.
For each variable $x_i \in V(\phi)$, 
we add three vertices $x_i, \bar x_i, y_i$ to $V_\phi$.

We next define $\mathcal A_\phi$.
For each $x_i \in V(\phi)$, 
we add 
$(\set{x_i}, \set{y_i})$, 
$(\set{\bar x_i}, \set{y_i})$, 
$(\set{s}, \set{x_i})$, 
and $(\set{s}, \set{\bar x_i})$.
For each clause $C_j = \ell^j_1 \lor \ell^j_2 \lor \ell^j_3$, 
we add a hyperarc $(\set{\bar \ell^j_1, \bar \ell^j_2, \bar \ell^j_3}, \set{c})$.
Finally, we add a hyperarc $(\set{c, y_1, \ldots, y_{n}}, \set{t})$, where $n$ is the number of variables in $\phi$.

We next define a set of minimal $s$-$t$ separators $\mathcal X(\phi)$.
As a set of vertices $\set{c}$ or $\set{y_j}$ for any $1 \le j \le n$ is a minimal $s$-$t$ separator,
we add such trivial minimal separators to $\mathcal X(\phi)$.
For each variable $x_i \in V(\phi)$, $\set{x_i, \bar x_i}$ is also a minimal $s$-$t$ separator.
We define a set of minimal $s$-$t$ separators $\mathcal X(\phi)$ as the above $2n + 1$ minimal $s$-$t$ separators.
\Cref{fig:reduction-AMS} gives an example of our reduction.


\begin{figure}[t]
    \centering
    \includegraphics[width=\linewidth]{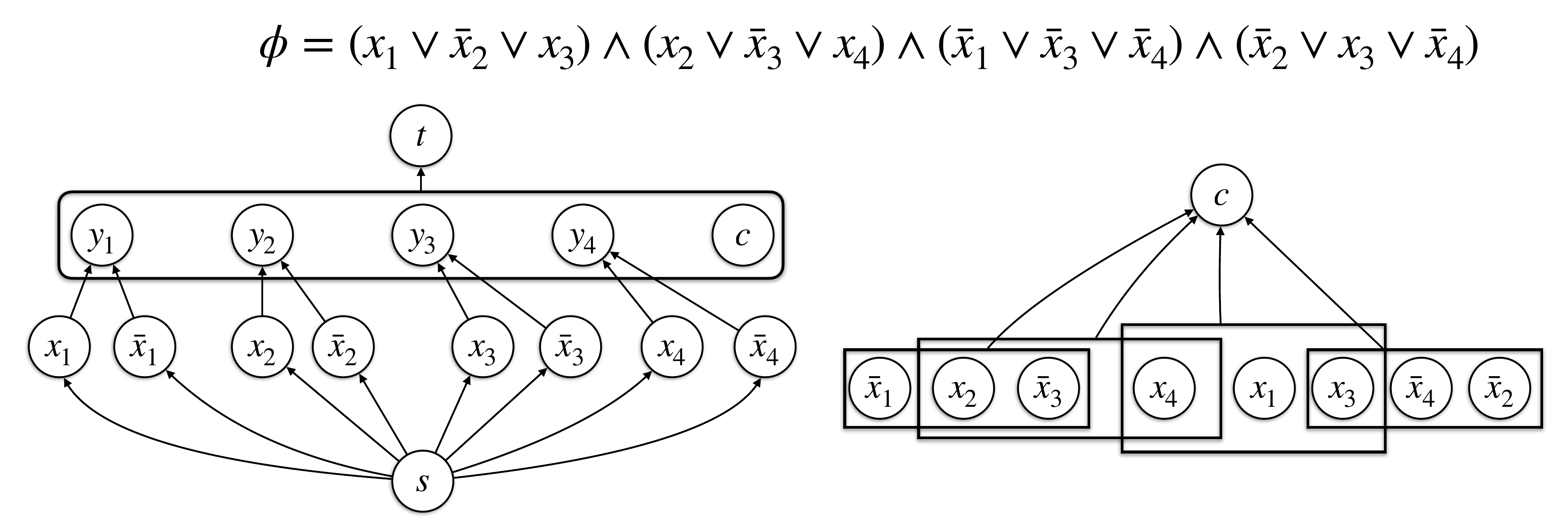}
    \caption{An example of our reduction from 3-\SAT{} to \textsc{Another Mimimal $s$-$t$ separator}. 
    The left figure shows the vertices and hyperarcs added for the variables of $\phi$, 
    while the right figure shows the hyperarcs added for the clauses.
    The entire directed hypergraph $\mathcal D_\phi$ is the union of the vertices and hyperarcs of these two directed hypergraphs.
    In this case, $\mathcal D_\phi$ has the following minimal $s$-$t$ separators:
    $\set{x_1, \bar x_1}$, 
    $\set{x_2, \bar x_2}$, 
    $\set{x_3, \bar x_3}$, 
    $\set{x_4, \bar x_4}$, 
    $\set{y_1}$, 
    $\set{y_2}$, 
    $\set{y_3}$, 
    $\set{y_4}$, and
    $\set{c}$. 
    Since $\phi$ has a satisfying assignment 
    $x_1 = 1$,
    $x_2 = 1$,
    $x_3 = 0$, and
    $x_4 = 0$,
    $\mathcal D_\phi$ has a minimal $s$-$t$ separator $\set{x_1, x_2, \bar x_4}$.
    If we assign $x_1 = 1$, $x_2 = 1$, and $x_4 = 0$, then $\phi$ is satisfied regardless of the assignment to $x_3$. 
    That is, $x_3$ becomes a ``don't care'' variable, and therefore neither $x_3$ nor $\bar{x}_3$ is included in the hyperpath.    
    }
    \label{fig:reduction-AMS}
\end{figure}

\begin{lemma}
    For a 3-CNF formula $\phi$,
    $\phi$ has a satisfying assignment if and only if
    $\mathcal D_\phi$ has a minimal $s$-$t$ separator not contained in $\mathcal X(\phi)$.
\end{lemma}
\begin{proof}
    Suppose that $\phi$ has a satisfying assignment $\alpha$.
    We define a set of vertices using $\alpha$.
    Let $x_i$ be a variable in $\phi$.
    The set of vertices $X$ contains a vertex $x_i$ if $\alpha(x_i) = 1$,
    otherwise $X$ contains $\bar x_i$.
    From the construction of $\mathcal D_\phi$,
    each $y_i$ is connected to $s$ in $\mathcal D_\phi[X]$.
    Moreover, since $\alpha$ is a satisfying assignment,
    $c$ is not connected to $s$ in $\mathcal D_\phi[X]$.
    Therefore, $X$ is an $s$-$t$ separator and 
    a minimal $s$-$t$ separator that is contained in $X$ is not contained in $\mathcal X(\phi)$.
    
    Suppose that $\mathcal D_\phi$ has a minimal $s$-$t$ separator $X \not\in \mathcal X(\phi)$.
    From the construction of $\mathcal X(\phi)$ and $X \not\in \mathcal X(\phi)$,
    for each $1 \le i \le \size{V(\phi)}$, $x_i \notin X$ or $\bar x_i \notin X$.
    We define an assignment $\alpha_X$ as follows.
    If $x_i \in X$, $\alpha_X(x_i) = 1$, otherwise $\alpha_X(x_i) = 0$.
    From the construction of $\mathcal D_\phi$, 
    each clause is satisfied. 
    Otherwise, $\mathcal D_\phi[V(\mathcal D_\phi) \setminus X]$ has an $s$-$t$ hyperpath and it contradicts that $X$ is an $s$-$t$ separator.
    Therefore, $\alpha_X$ is a satisfying assignment of $\phi$.
\end{proof}

We improve our reduction to the cardinality of each tail to two by making a small modification, 
In our reduction, the maximum cardinality of each tail is $n + 1$ since $\mathcal D_\phi$ has a hyperarc $A = (\set{y_1, \ldots, y_{n}, c}, \set{t})$.
In addition, $\mathcal D$ has $m$ hyperarcs whose tail has the cardinality three.
To simplify the following discussion, we assume that $n + m + 1$ is a power of two.
If $n+m+1$ is not a power of two, we add one dummy variable $z$ and dummy clauses $(z \lor z \lor z)$ to make $n+m+1$ a power of two.

For each clause $C_j = \ell^j_1 \lor \ell^j_2 \lor \ell^j_3$, we add one vertex $c_j$ and 
replace a hyperarc $(\set{\bar \ell^j_1, \bar \ell^j_2, \bar \ell^j_3}, \set{c})$ with two hyperarcs 
$(\set{\bar \ell^j_1, \bar \ell^j_2}, \set{c_j})$ and
$(\set{c_j, \bar \ell^j_3}, \set{c})$.
We replace a hyperarc $(\set{c, y_1, \ldots, y_n}, \set{t})$ with $A = (\set{c, c_1, \ldots, c_m, y_1, \ldots, y_n}, \set{t})$.
As a result of this modification, all hyperarcs except one have tails with cardinality at most two.
This modification adds only $\set{c_j}$ for each $1 \le j \le m$ as new minimal $s$-$t$ separators since $\set{c_j}$ is a minimal $s$-$t$ separator.
Applying the similar technique as in the previous section, 
we add vertices $d_2, \ldots, d_{n + m}$, and
replace a hyperarc $A = (\set{y_1, \ldots, y_n, c_1, \ldots, c_m, c}, \set{t})$ with
a set of hyperarcs $\inset{(\set{w_{2i}, w_{2i+1}}, \set{w_{i}})}{1 \le i \le n+m}$, 
where $w_1 = t$, $w_i = d_i$ if $i$ is at least $2$ and at most $n + m$,
$w_{n + m + 1} = c$, $w_i = c_{i - n - m - 1}$ if $i$ is at least $n + m + 2$ and at most $n + 2m + 1$, otherwise, $w_i = y_{i - n - 2m - 1}$.
Applying this modification, we obtain new $\order{n + m}$ minimal $s$-$t$ separators.
Therefore, by adding such small minimal separators to $\mathcal X(\phi)$,
the set of minimal separators that is not contained in $\mathcal X(\phi)$ does not change.


\begin{theorem}
    \textsc{Another Minimal $s$-$t$ Separator} on a $B$-hypergraph $\mathcal D = (V, \mathcal A)$ is \NP-complete even if $\size{\mathcal X} \le \size{V}$ and $\size{T(A)} \le 2$ for any $A \in \mathcal A$.
    Therefore, there are no output-polynomial time algorithms for enumerating all minimal $s$-$t$ separators in a $B$-hypergraph unless $\P = \NP$ even if $\size{T(A)} \le 2$ holds for any $A \in \mathcal A$.    
\end{theorem}

\subsection{\texorpdfstring{$s$-$t$}{s-t} hyperpath Enumeration}\label{subsec:bpath:bf}
Let $\mathcal H = (V, \mathcal E)$ be a hypergraph.
We construct a directed hypergraph $\mathcal D_{\mathcal H} = (V_{\mathcal H}, \mathcal A_\mathcal H)$ as follows.
The set of vertices $V_{\mathcal H}$ consists of $\mathcal E \cup \set{s, t}$.
For each $v \in V$, we add a hyperarc $(\set{s}, \mathcal E_v)$, 
where $\mathcal E_v$ is the set of hyperedges in $\mathcal E$ that contain $v$.
Note that $\mathcal E_v$ is the set of hyperedges in $\mathcal H$, but
it is the set of vertices in $\mathcal D_{\mathcal H}$.
Thus, $(\set{s}, \mathcal E_v)$ is a pair of distinct sets of vertices in $V_{\mathcal H}$.
Additionally, we add hyperarc $(\mathcal E, \set{t})$.
Since each hyperarc $A \in \mathcal A_\mathcal H$ satisfies either $\size{T(A)} = 1$ or $\size{H(A)} = 1$,
$\mathcal D_\mathcal H$ is a $BF$-hypergraph.

\begin{lemma}\label{lem:bij:bf}
    Let $\mathcal H$ be a hypergraph and $\mathcal D_{\mathcal H}$ be a $BF$-hypergraph.
    There is a bijection between 
    the set of minimal transversals in $\mathcal H$ and the set of $s$-$t$ hyperpaths in $\mathcal D_\mathcal H$.    
\end{lemma}
\begin{proof}

 Let $\mathcal H$ be a hypergraph and $\mathcal D_{\mathcal H}$ be a $BF$-hypergraph obtained by the above procedure.
    We define a function $f: \mathrm{Tr}(\mathcal H) \to 2^{\mathcal A_{\mathcal H}}$ such that $f(T) = \inset{(\set{s}, \mathcal E_v)}{v \in T} \cup \set{(\mathcal E, \set{t})}$, where $T$ is a minimal transversal of $\mathcal H$.

    We show that for a minimal transversal $T$, $f(T)$ is an $s$-$t$ hyperpath.
    To this end, we show that $t$ is $B$-connected from $s$ in $\mathcal D_\mathcal H[f(T)]$.
    Since $T$ is a transversal, any vertex in $V(\mathcal D_{\mathcal H}) \setminus \set{s, t}$ is connected to $s$.
    From the definition of $f(T)$, $f(T)$ contains a hyperarc $(\mathcal E, \set{t})$.
    Since any vertex in $V(\mathcal D_{\mathcal H}) \setminus \set{s, t}$ is connected to $s$ and $f(T)$ contains $(\mathcal E, \set{t})$, $t$ is connected to $s$.
    We next show the minimality of $f(T)$ by contradiction.
    Suppose that there is an $s$-$t$ hyperpath $P_{st}$ contained in $f(T)$.
    If $(\mathcal E, \set{t}) \not\in P_{st}$, $P_{st}$ is not $s$-$t$ hyperpath.
    Thus, a hyperarc $A \in f(T) \setminus P_{st}$ satisfies $T(A) = \set{s}$ and $H(A) = \mathcal E_v$ for some $v \in V(\mathcal H)$.
    However, since $T$ is a minimal transversal, 
    there is a hyperedge $E \in \mathcal E_v$ such that $T \cap E = \set{v}$.
    In this case, $E$ is not connected to $s$.
    Therefore, $f(T)$ is an $s$-$t$ hyperpath.

    We next show that $f$ is a bijection between $\mathrm{Tr}(\mathcal H)$ and the set of all $s$-$t$ hyperpaths in $\mathcal D_\mathcal H$.
    We first show that $f$ is injective.
    Let $T_1$ and $T_2$ be two distinct minimal transversals.
    Since $T_1$ and $T_2$ are distinct, there are vertices 
    $v_1 \in T_1 \setminus T_2$ and
    $v_2 \in T_2 \setminus T_1$.
    Therefore,
    $f(T_1) \setminus f(T_2)$ contains a hyperarc $(\set{s}, \mathcal E_{v_1})$ and
    $f(T_2) \setminus f(T_1)$ contains a hyperarc $(\set{s}, \mathcal E_{v_2})$, and
    $f$ is injective.

    We next show that $f$ is surjective.
    Let $P_{st}$ be an $s$-$t$ hyperpath in $\mathcal D_{\mathcal H}$.
    From the construction of $\mathcal D_{\mathcal H}$,
    each hyperarc $A$ satisfying $T(A) = \set{s}$ corresponds to a vertex in $\mathcal H$.
    Thus, we consider a set of vertices $T$ corresponds to $P_{st} \setminus \set{(\mathcal E, \set{t})}$.
    Since $P_{st}$ is an $s$-$t$ hyperpath, each vertex $v \in T$ has a hyperedge $E$ such that $T \cap E = \set{v}$.
    Therefore, $T$ is a minimal transversal and $f$ is a bijection.
\end{proof}

\begin{thmrep}
    If there is an output-polynomial time algorithm for \textsc{$s$-$t$ hyperpath Enumeration} on a $BF$-hypergraph $\mathcal D$, 
    there is an output-polynomial time algorithm for \textsc{Minimal Transversal Enumeration}.
\end{thmrep}
\begin{proof}
    Suppose that there is an output-polynomial time algorithm for \textsc{$s$-$t$ hyperpath Enumeration} on a $BF$-hypergraph.
    For a hypergraph $\mathcal H = (V, \mathcal E)$, we construct $\mathcal D_{\mathcal H}$ in
    $\order{||\mathcal H||}$ time, where $||\mathcal H||$ is the sum of the number of vertices and the cardinality of hyperedges.
    Since we have an output-polynomial time algorithm for enumerating all $s$-$t$ hyperpaths in $\mathcal D_{\mathcal H}$,
    we obtain the set of all $s$-$t$ hyperpaths in $\mathcal D_{\mathcal H}$.
    From \Cref{lem:bij:bf}, there is a bijection between the set of minimal transversals in $\mathcal H$ and the set of all $s$-$t$ hyperpaths in $\mathcal D_{\mathcal H}$.
    Moreover, we can restore $\mathrm{Tr}(\mathcal H)$ in $\order{||\mathcal H|| \cdot \mathrm{Tr}(\mathcal H)}$ time.
    Therefore, we obtain an output-polynomial time algorithm for \textsc{Minimal Transversal Enumeration}.
\end{proof}

\section{\texorpdfstring{Polynomial-delay Enumeration of $s$-$t$ hyperpaths for $B$-hypergraphs}{Polynomial-delay Enumeration of s-t hyperpaths for B-hypergraphs}}
We give a positive result of the $s$-$t$ hyperpath enumeration for a $B$-hypergraph.
In this section, a directed hypergraph may have parallel hyperarcs, that is, there may be hyperarcs with the same pair of tail set and head.
The proposed algorithm solves a slightly more general problem than the $s$-$t$ hyperpath enumeration.
We first extend the concept of an $s$-$t$ hyperpath.

Let $S$ be a set of vertices.
We define \emph{$B$-connection from $S$} as follows:
(i)  a vertex $s \in S$ is $B$-connected from $S$ and
(ii) if there is a hyperarc $A$ such that all the vertices in $T(A)$ are $B$-connected from $S$, then $H(A)$ is $B$-connected from $S$.
For a set of vertices $S$ and $T$,
an inclusion-wise minimal set of hyperarcs $P_{ST}$ is an \emph{$S$-$T$ hyperpath} if 
all vertices
in $T$ is $B$-connected from $S$.
It should be noted that a $S$-$T$ hyperpath is also a generalization of 
a minimal directed Steiner tree~\cite{10.1145/3517804.3524148}.
Hereafter, we give a polynomial-delay enumeration algorithm for enumerating all $S$-$T$ hyperpaths.
We first introduce a characterization for an $S$-$T$ hyperpath in a $B$-hypergraph.


\begin{lemma}\label{obs:hpath}
    A set of hyperarcs $P_{ST}$ is an $S$-$T$ hyperpath if and only if the followings hold:
    \begin{enumerate}
        \item There is a sequence of hyperarcs $(A_1, \ldots, A_{\size{P_{ST}}})$ such that for each $A_i$, $T(A_i) \subseteq S \cup \bigcup_{1 \le j \le i-1} H(A_j)$,
        \item $|\{A\in P_{ST} \mid \set{v} = H(A)\}| = 0$ if $v \in S$, 
        otherwise $|\{A\in P_{ST} \mid \set{v} = H(A)\}| = 1$\label{item:minimal_BsT_to_much_in}, and
        \item  $|\{A\in P_{ST} \mid v \in T(A)\}| \ge 1$ if $v \in V(\mathcal D[P_{ST}]) \setminus (T \cup S)$.\label{item:minimal:hyperarcs}
    \end{enumerate}
\end{lemma}
\begin{proof}
    Suppose that $P_{ST}$ satisfies the three conditions and $(A_1, \ldots, A_{\size{P_{ST}}})$ is an ordering of $P_{ST}$.
    From the condition~1 and 2, 
    all vertices in $\mathcal D[P_{ST}]$ is $B$-connected from $S$.
    We show that the minimality of $P_{ST}$ by contradiction.
    Suppose that $P'_{ST} \subset P_{ST}$ is an $S$-$T$ hyperpath.
    Let $A$ be the largest hyperarc in $P'_{ST} \setminus P_{ST} $ with respect to the order of $P_{ST}$.
    If the head of $A$ is contained in $T$, 
    $t$ is not $B$-connected from $S$ since $P'_{ST}$ has no hyperarcs with $t$ as its head.
    Thus, the head of $A$ is not in $T$.
    If the head of $A$ is contained in $S$, it contradicts the condition~2.
    Thus, it is contained in $V \setminus (S \cup T)$.
    From the condition~3, $P_{ST}$ contains a hyperarc $A'$ that contains the head of $A$ as its tail.
    From the condition~1, $A'$ is larger than $A$ with respect to the order $(A_1, \ldots, A_{\size{P_{ST}}})$.
    If $A'$ is contained in $P'_{ST}$, it contradicts the minimality of $P'_{ST}$ since removing $A'$ from $P'_{ST}$ does not change $B$-connectivity for other vertices.
    If $A'$ is not contained in $P'_{ST}$,
    it contradicts the condition of $A$ since $A'$ is larger than $A$ with respect to the order of $P_{ST}$, and
    $P_{ST}$ is an $S$-$T$ hyperpath.

    Suppose that $P_{ST}$ is an $S$-$T$ hyperpath.
    Since $P_{ST}$ is an $S$-$T$ hyperpath, 
    the conditions~2 and 3 are satisfied, otherwise, it contradicts the minimality of $P_{ST}$.
    Therefore, we  show that $P_{ST}$ has an order satisfying the condition~1.
    We order $P_{ST}$ as follows.    
    Let $\mathcal A_1$ be the set of hyperarcs such that $T(A) \subseteq S$.
    We recursively define $\mathcal A_j$ as $\inset{A \in P_{ST} \setminus \mathcal A_{i-1}}{T(A) \subseteq S \cup \bigcup_{B \in \mathcal A_{i-1}}H(B)}$.
    Notice that the set of vertices $S \cup \bigcup_{B \in \mathcal A_{i}}H(B)$ is $B$-connected from $S$ in a directed subhypergraph induced by $\bigcup_{1 \le j \le i} \mathcal A_j$.
    If an hyperarc $A \in P_{ST}$ is not contained any $\mathcal A_i$, it contradicts the minimality of $P_{ST}$.
    Moreover, for distinct $k_1$ and $k_2$, 
    $\mathcal A_{k_1}$ and $\mathcal A_{k_2}$ are disjoint.
    Therefore, $\set{\mathcal A_1, \ldots, \mathcal A_\ell}$ gives a partition of $P_{ST}$, where $\ell$ is the maximum integer satisfying $\mathcal A_\ell \neq \emptyset$.
    We consider an order of $P_{ST}$ as $(\mathcal A_1, \ldots, \mathcal A_\ell)$.
    We obtain a desired order of the set of hyperarcs in $P_{ST}$ 
    by ordering the hyperarcs within each $\mathcal A_i$ in an arbitrary order.
\end{proof}    

Our goal is to enumerate all sets of hyperarcs that satisfy the conditions in \Cref{obs:hpath} by a backtracking approach.
To this end, we recursively partition the problem into two subproblems.
This approach is inspired by Read and Tarjan's $s$-$t$ path enumeration algorithm~\cite{ReadTarjan}.
We describe an overview of our algorithm.

In our enumeration algorithm, 
we find an $S$-$T$ hyperpath $P_{ST}$ with an ordering $(A_1, \ldots, A_k = (T_k, h_k))$ in \Cref{obs:hpath}, 
divide enumeration problems into the enumeration of an $S$-$T$ hyperpath that contains $A_{k}$ and does not contain $A_{k}$.
The $S$-$T$ hyperpaths outputted by these two enumeration problems are clearly non-duplicate, 
and their union is equal to the set of $S$-$T$ hyperpaths of the original problem.
One subproblem, enumeration of $S$-$T$ hyperpaths that do not contain $A_{k}$ equals enumeration of $S$-$T$ hyperpaths in $D[\mathcal A \setminus \set{A_{k}}]$. 
The key observation is that the enumeration of $S$-$T$ hyperpaths can also be achieved by modifying $\mathcal D$ and $T$.
We give the details of our algorithm in \Cref{alg:proposed}.

\begin{algorithm}[t]
    \SetAlgoLined
    \SetKwProg{Procedure}{Procedure}{}{}
    \SetKwFunction{EnumHyperpaths}{EnumHyperpaths}
    \DontPrintSemicolon
    \Procedure{\EnumHyperpaths{$\mathcal D = (V, \mathcal A), S, T, P$}}{
        \lIf{$T = \emptyset$}{Output $P$ and \Return}
        \lIf{$\mathcal D$ has no $S$-$T$ hyperpaths}{\Return}
        Let $(A_1, \ldots, A_k = (T_k, h_k))$ be an $S$-$T$ hyperpath.\;
        \EnumHyperpaths{$\mathcal D[\mathcal A \setminus \set{A_k}], S, T, P$}\;
        $\mathcal A', T' \gets \emptyset, (T \cup T_k) \setminus (S \cup \set{h_k})$\tcp*{Note that $\mathcal A'$ is a multiset.}\label{alg:rec:0}
        \ForEach{$B \in \mathcal A \setminus \inset{A'' \in \mathcal A}{H(A'') = h_k}$}{
            \lIf{$h_k \in T(B)$}{
                Add $((T(B) \setminus \set{h_k}) \cup T_k, H(B))$ to $\mathcal A'$\label{alg:diff}.
            }\lElse{
                Add $B$ to $\mathcal A'$.
            }
            }
        
        \EnumHyperpaths{$(V \setminus \set{h_k}, \mathcal A'), S, T', P \cup \set{A_k}$}\label{alg:rec:1}\;        
    }
    \caption{A polynomial-delay and polynomial-space algorithm for enumerating all $S$-$T$ hyperpaths.}
    \label{alg:proposed}
\end{algorithm}


Hereafter, we formally define a division of the problem.
Let $\mathcal D = (V, \mathcal A)$ be a directed hypergraph and $P_{ST}$ be an $S$-$T$ hyperpath and $(A_1, \ldots, A_k = (T_k, h_k))$ be an order of $P_{ST}$ satisfying the conditions in \Cref{obs:hpath}.
We define $\mathcal P_{ST}(\mathcal D)$ as the set of $S$-$T$ hyperpaths in $\mathcal D$.
Furthermore, we introduce $\mathcal P^0_{ST}(\mathcal D, P_{ST}, A_k)$ and $\mathcal P^1_{ST}(\mathcal D, P_{ST}, A_k)$ as the set of $S$-$T$ hyperpaths that do not contain $A_k$ and contain $A_k$, respectively.
Obviously, the set of all $S$-$T$ hyperpaths on $\mathcal D[\mathcal A \setminus \set{A_k}]$ corresponds to $\mathcal P^0_{ST}(\mathcal D, P_{ST}, A_k)$.
The remainder part is how to obtain $\mathcal P^1_{ST}(\mathcal D, P_{ST}, A_k)$ by modifying $\mathcal D$, $S$, and $T$.

Since any $S$-$T$ hyperpath in $\mathcal P^1_{ST}(\mathcal D, P_{ST}, A_k)$ contains $A_k$, 
there are no $S$-$T$ hyperpaths that contain a hyperarc whose head is $h_k$ without $A_k$ from \Cref{obs:hpath}.
Thus, removing all hyperarcs that contain $h_k$ as its head from $\mathcal D$ does not change the set of $S$-$T$ hyperpaths containing $A_k$.
Moreover, $u \in T_k$ is $B$-connected from $S$ in any $S$-$T$ hyperpath in $\mathcal P^1_{ST}(\mathcal D, P_{ST}, A_k)$ from \Cref{obs:hpath}.
Thus, for any $Q \in \mathcal P^1_{ST}(\mathcal D, P_{ST}, A_k)$ and $B \in Q$ such that $T(B)$ contains $h_k$,
any vertex in $(T(B) \setminus \set{h_k}) \cup T_k$ is $B$-connected from $S$ in $\mathcal D[Q]$.
Motivated by these observations, we consider the following modification of hyperarcs in $\mathcal A \setminus \inset{A'' \in \mathcal A}{H(A) = h_k}$.
For a hyperarc $B \in \mathcal A \setminus \inset{A'' \in \mathcal A}{H(A'') = h_k}$,
we define a function $f$ as follows.
If $B$ does not contain $h_k$ as its tail, $f(B) = B$, otherwise, 
$f(B) = ((T(B) \setminus \set{h_k}) \cup T_k, H(B))$.
Based on $f$, we define the multiset of hyperarcs as $\mathcal B = \inset{f(B)}{B \in \mathcal A \setminus \inset{A'' \in \mathcal A}{H(A'') = h_k}}$.
Note that $f$ is a bijection between $\mathcal A \setminus \inset{A'' \in \mathcal A}{H(A'') = h_k}$ and $\mathcal B$.
We show that there is a bijection between the set of all $S$-$T'$ hyperpaths on $\mathcal D'$ and $\mathcal P^1_{ST}(\mathcal D, P_{ST}, A_k)$, 
where $\mathcal D' = (V \setminus \set{h_k}, \mathcal B)$ and $T' = (T \cup T(A_k)) \setminus (S \cup \set{H(A_k)})$.

\begin{lemma}\label{lem:bij_Y_A}
    There is a bijection between $\mathcal P_{ST'}(\mathcal D')$ and $\mathcal P^1_{ST}(\mathcal D, P_{ST}, A_k)$, 
    where $\mathcal D' = (V \setminus \set{h_k}, \mathcal B)$ and $T' = (T \cup T(A_k)) \setminus (S \cup \set{H(A_k)})$.
\end{lemma}
\begin{proof}   
    For each $Q \in \mathcal P^1_{ST}(\mathcal D, P_{ST}, A_k)$, 
    we define the set of hyperarcs $g(Q) = \bigcup_{B \in Q\setminus\set{A_k}}f(B)$.
    We show that $g$ is a bijection between $\mathcal P^1_{ST}(\mathcal D, P_{ST}, A_k)$ and $\mathcal P_{ST'}(\mathcal D')$.
    We first show that $g$ is injective.
    Let $Q_1$ and $Q_2$ be distinct $S$-$T$ hyperpaths in $\mathcal P^1_{ST}(\mathcal D, P_{ST}, A_k)$.
    Thus, 
    both $Q_1 \setminus Q_2$ and
    $Q_2 \setminus Q_1$ are non-empty.
    It implies that $g(Q_1)$ and $g(Q_2)$ are distinct.

    We next show that $g$ is surjective.
    Let $Q'$ be an $S$-$T'$ hyperpath in $\mathcal P_{ST'}(\mathcal D')$.
    Since $Q'$ is an $S$-$T'$ hyperpath, any vertex in $T'$ is $B$-connected from $S$.
    Thus, $h_k$ is $B$-connected from $S$ in $\mathcal D[Q' \cup \set{A_k}]$ since $T_k \subseteq S \cup T'$.
    Moreover, any vertex in $\mathcal D[\set{A_k} \cup \bigcup_{B' \in Q'}f^{-1}(B')]$ is $B$-connected from $S$
    since any vertex in $T_k$ is $B$-connected from $S$ in $\mathcal D[\set{A_k} \cup \bigcup_{B' \in Q'}f^{-1}(B')]$.
    Finally, if $\set{A_k} \cup \bigcup_{B' \in Q'}f^{-1}(B')$ is not an $S$-$T$ hyperpath, then it contradicts the minimality of $Q'$.
    Therefore, $\set{A_k} \cup \bigcup_{B' \in Q'}f^{-1}(B')$ is an $S$-$T$ hyperpath in $\mathcal P^1_{ST}(\mathcal D, P_{ST}, A_k)$, and
    $g$ is surjective.
\end{proof}

\Cref{lem:bij_Y_A} ensures that our branching strategy correctly partitions the set of $S$-$T$ hyperpaths in $\mathcal D$.
Moreover, for each $S$-$T'$ hyperpath $P'_{ST'}$ in $\mathcal D'$, 
it is easy to restore the $S$-$T$ hyperpath $P$ that corresponds to $P'$.
Therefore, by recursively applying this partitioning, we can enumerate all $S$-$T$ hyperpaths, and the following theorem holds.

\begin{thmrep}
    \Cref{alg:proposed} enumerates all $S$-$T$ hyperpaths with $\order{m^2 ||\mathcal A||}$ delay and $\order{m||\mathcal A||}$ space, where
    $m = \size{\mathcal A}$ and
    $||\mathcal A|| = \sum_{A \in \mathcal A}(|H(A)| + |T(A)|)$.
\end{thmrep}

\begin{proof}
    The correctness follows from \Cref{lem:bij_Y_A}, and
    the space complexity is bounded by $\order{m ||\mathcal A||}$
    since the depth of this recursion tree is bounded by $m$, 
    Thus, we give the time complexity analysis.


    In each recursion step, we find an $S$-$T$ hyperpath.
    It can be done in $\order{m ||A||}$ time using $B$-connectivity checking algorithm in \cite{AUSIELLO2017293}.
    In addition, finding $P_{ST}$ is a bottleneck of each recursion procedure without line~\ref{alg:rec:0} and \ref{alg:rec:1}.
    As the depth of this recursion tree is at most $m$, 
    the delay of this algorithm is $\order{m^2 ||A||}$.
\end{proof}

\section{Conclusion}

In this paper, we show that there are no output-polynomial time algorithms for induced $s$-$t$ hyperpath enumeration and minimal $s$-$t$ separator enumeration in a $B$-hypergraph
even if the cardinality of tails is constant, unless $\P=\NP$.
Moreover, enumerating $s$-$t$ hyperpaths in a $BF$-hypergraph is at least as hard as the problem of enumerating the minimal transversals of a hypergraph.
This indicates that
enumeration of $s$-$t$ hyperpaths in a directed hypergraph in output-polynomial time is a challenging problem.
Finally, we give a polynomial-delay algorithm for enumerating all $S$-$T$ hyperpaths.
This algorithm is based on a simple backtracking algorithm inspired by \cite{Read:Tarjan:Networks:1975}.

The remaining problem is the enumeration of minimal $s$-$t$ cuts in a directed hypergraph.
This problem is equivalent to enumerating all maximal signatures/minimal unsatisfiable subformula in Horn~\SAT{} formulae if an input directed hypergraph is a $B$-hypergraph.
Recently, enumerating (maximal) signatures in tractable \SAT{} formulae have been studied~\cite{BERCZI202168,DBLP:journals/corr/abs-2402-18537}.
Nadia et al. showed that in OXR-\SAT{} case, this problem can be solved in incremental polynomial time.

As interesting directions for future research, we consider the improvement of the delay in enumerating $S$-$T$ hyperpaths and 
the existence of an output-quasi polynomial time enumeration algorithm, that is, an enumeration algorithm that runs in $(n + N)^{\polylog {(n + N)}}$ time, for $S$-$T$ hyperpaths in general directed hypergraphs.
The enumeration of $s$-$t$ paths in undirected graphs is a well-studied topic, and certain optimal algorithms are already known~\cite{Birmele:Ferreira:SODA:2013}. 
Furthermore, for the enumeration of minimal Steiner trees, which generalizes directed $s$-$t$ paths, a linear-delay algorithm has been known~\cite{10.1145/3517804.3524148}.
Investigating how much the delay can be improved remains an interesting open problem.

In addition, it remains an open question whether the $S$-$T$ hyperpath enumeration can also be solved in output-quasi polynomial time 
since the minimal transversal enumeration can be solved in output-quasi-polynomial time.
It would also be interesting to show whether the corresponding another solution problem of \textsc{$s$-$t$ hyperpath Enumeration} is \NP-complete.
Conducting more detailed research on this question would be an interesting direction for future work.


\bibliographystyle{abbrv}
\bibliography{main.bib}

\end{document}